%% file: main.tex
\documentclass[journal]{IEEEtran}

\usepackage[T1]{fontenc}
\usepackage{ucs}
\usepackage[utf8x]{inputenc}
  
\usepackage[greek,english]{babel}
\usepackage[babel]{csquotes}

\usepackage[noadjust]{cite}

\usepackage{fixmath}
\usepackage{mathtools}
\usepackage{amsmath}
\usepackage{amsthm}
\usepackage{amssymb}
\usepackage{graphicx}
\usepackage{tikz}
\usepackage{pgfplots}
\usepackage{algorithm}
\usepackage{algorithmic}
\usepackage{stfloats}
\usepackage{environ}
\usepackage{subeqnarray}

\newcommand{\tmpand}{&}

\ifCLASSOPTIONdraftcls
\NewEnviron{mueq}{\begin{equation}\BODY\end{equation}}
\else
\NewEnviron{mueq}{\begin{multline}\BODY\end{multline}}
\fi

\DeclareMathAlphabet{\mathbit}{OML}{cmr}{bx}{it}
\DeclareMathAlphabet{\mathss}{T1}{lmss}{m}{it}
\DeclareMathAlphabet{\mathssbold}{T1}{lmss}{bx}{sl}
\DeclareMathAlphabet{\mathssgreek}{LGR}{lmss}{m}{sl}
\DeclareMathAlphabet{\mathssgreekbold}{LGR}{lmss}{bx}{sl}

\DeclareSymbolFont{sssymbols}{T1}{lmss}{m}{it}
\SetSymbolFont{sssymbols}{bold}{T1}{lmss}{bx}{it}
\DeclareMathAccent{\rtilde}{\mathalpha}{sssymbols}{"03}

\newcommand{\mbc}[1]{\mathssbold{#1}}
\newcommand{\mc}[1]{\mathss{#1}}

\newcommand{\mcg}[1]{\mathssgreek{#1}}
\newcommand{\mb}[1]{\mathbit{#1}}
\newcommand{\mt}[1]{\mathrm{#1}}  

\DeclareMathOperator{\TransposedOp}{T}

\newcommand{\crr}[1]{{\check{\mb #1}}}	

\newcommand{\one}{\boldsymbol{1}} 
\newcommand{\zero}{\boldsymbol{0}} 
\newcommand{\id}{\mathbf{I}} 

\newcommand{\J}{\mathrm{j}}
\newcommand{\E}{\mathrm{e}}

\newcommand{\EntrOp}{\operatorname{h}}
\newcommand{\Entr}[1]{\EntrOp\left(#1\right)}

\newcommand{\tr}[1]{\operatorname{tr}\lbrack#1\rbrack}	
\DeclareMathOperator{\st}{s.t.}

\newcommand{\Expect}[1]{\operatorname{E}\lbrack#1\rbrack}

\newcommand{\Tr}{{\TransposedOp}}
\newcommand{\bTr}{{\TransposedOp}}

\DeclareMathOperator*{\argmax}{argmax}

\newcommand{\allk}{{\forall k}}
\newcommand{\Ms}{L}
\newcommand{\ms}{\ell}
\newcommand{\Lag}{\Theta}

\newcommand{\brbfun}{F}
\newcommand{\brborigfun}{f}
\newcommand{\brbinitfun}{\hat f}
\newcommand{\brbeps}{\epsilon}
\newcommand{\brbBset}{\mathbb{B}}
\newcommand{\brbB}{\mathcal{B}}
\newcommand{\brbBox}[2]{\left[{#1};~{#2}\right]}
\newcommand{\brbBoxSmall}[2]{[{#1};~{#2}]}
\newcommand{\brbUsymb}{U}
\newcommand{\brbLsymb}{A}
\newcommand{\brbU}[1][\brbBox{\brba}{\brbb}]{\brbUsymb(#1)}
\newcommand{\brbL}[1][\brbBox{\brba}{\brbb}]{\brbLsymb(#1)}
\newcommand{\brbxsymb}{x}
\newcommand{\brbysymb}{y}
\newcommand{\brbasymb}{a}
\newcommand{\brbbsymb}{b}
\newcommand{\brbpsymb}{p}
\newcommand{\brbx}{\mb\brbxsymb}
\newcommand{\brbxk}[1][k]{\brbxsymb_{#1}}
\newcommand{\brby}{\mb\brbysymb}
\newcommand{\brbyk}[1][k]{\brbysymb_{#1}}
\newcommand{\brba}{\mb\brbasymb}

\newcommand{\hatbrbak}[1][k]{{\hat\brbasymb}_{#1}}
\newcommand{\brbb}{\mb\brbbsymb}

\newcommand{\hatbrbbk}[1][k]{{\hat\brbbsymb}_{#1}}
\newcommand{\brbp}{\mb\brbpsymb}
\newcommand{\brbpk}[1][k]{\brbpsymb_{#1}}

\newcommand{\cpeps}{\epsilon_\text{CP}}

\newcommand{\noise}{\mcg h}
\newcommand{\pow}{p}

\theoremstyle{remark}
\newtheorem{theorem}{Theorem}
\newtheorem{lemma}{Lemma}

\begin{document}

\title{Improper Signaling versus Time-Sharing in the Two-User Gaussian Interference Channel with TIN}

\author{Christoph~Hellings,~\IEEEmembership{Member,~IEEE}, 
        and~Wolfgang~Utschick,~\IEEEmembership{Senior~Member,~IEEE}
\thanks{The authors are with Technische Universit\"at M\"unchen, Professur f\"ur Methoden der Signalverarbeitung, 80290 M\"unchen, Germany, 
Telephone: +49 89 289-28520, e-mail: hellings@tum.de, utschick@tum.de. C. Hellings is now with Department of Physics, ETH Zurich, 8093 Zurich, Switzerland.

This paper was presented in part (algorithmic aspects and numerical results) at the 22nd International ITG Workshop on Smart Antennas (WSA 2018) \cite{HeUt18}. The main result (optimality of proper signaling) is a novel contribution of this paper.

Copyright (c) 2020 IEEE. Personal use of this material is permitted.  However, permission to use this material for any other purposes must be obtained from the IEEE by sending a request to pubs-permissions@ieee.org.}}

\maketitle

\begin{abstract}
So-called improper complex signals have been shown to be beneficial in the single-antenna two-user Gaussian interference channel 
under the assumptions that all input signals are Gaussian and that we treat interference as noise (TIN).
This result has been obtained under a restriction to pure strategies without time-sharing,
and it was extended to the case where the rates, but not the transmit powers, may be averaged over several transmit strategies.
In this paper, we drop such restrictions and discuss the most general case of coded time-sharing, where both the rates and the powers may be averaged.
Since coded time-sharing can in general not be expressed by means of a convex hull of the rate region,
we have to account for the possibility of time-sharing already during the optimization of the transmit strategy.
By means of a novel channel enhancement argument, we prove a surprising result:
proper signals are optimal if coded time-sharing is allowed.
In addition to establishing this result, we present an algorithm to compute the corresponding achievable rate region.
\end{abstract}


\section{Introduction}
\label{sec:intro}
While proper Gaussian signals are the optimal input signals in single-user systems with Gaussian noise,
improper input signals can be necessary to exploit the full potential of multiuser systems with interference.
The term \emph{proper} means in this context that the so-called pseudovariance $\mc{\rtilde c}_{\mc x}=\Expect{(\mc x-\Expect{\mc x})^2}$ 
of a complex random variable $\mc x$ is zero \cite{NeMa93}.\footnote{In the case of zero-mean Gaussian random variables, propriety is equivalent to circular symmetry of the probability density function.}

For the three-user interference channel, it was shown in \cite{CaJaWa10} that the optimal degrees of freedom (DoF)
can in general only be achieved using improper transmit signals, i.e., using signals with nonzero pseudovariance.
This result was based on interference alignment and is thus specific to systems with three or more users.
However, it has inspired researchers to also consider the use of improper signals in two-user interference channels.

In \cite{HoJo12}, the Gaussian two-user single-antenna interference channel was studied from a game-theoretic perspective.
Under the assumption that Gaussian codebooks are used and that the receivers treat interference as noise (TIN),
it was shown that a cooperative solution based on improper signaling can outperform the Nash equilibrium obtained with proper input signals.
Moreover, a parametrization of the Pareto boundary of the achievable rate region was given for the special case of maximally improper signals
(corresponding to rank-one beamforming in an equivalent real-valued system).
An algorithm for signal-to-interference-and-noise ratio (SINR) balancing with maximally improper signals and zero-forcing was developed in \cite{PaPaKiLe13},
and \cite{ZeYeGuGuZh13} proposed two algorithms for suboptimal rate balancing with general improper signals based on 
semidefinite relaxation and based on a two-stage method that optimizes the transmit power and the impropriety of each transmit signal in two separate steps.
As extensions to multiantenna systems, \cite{ZeZhGuGu13} proposed a method for rate balancing in the multiple-input single-output (MISO) interference channel
based on a similar two-stage approach,
and \cite{LaAgVi16} proposed to obtain suboptimal solutions 
to a weighted sum rate maximization 
in the multiple-input multiple-output (MIMO) interference channel
via a weighted MSE formulation based on \cite{ChAgCaCi08}.

The optimization of improper signaling strategies in the interference channel is a nonconvex problem,
and the approaches presented above do not guarantee globally optimal solutions.
However, for the case of a single-antenna system, Pareto-optimal transmit strategies for the case of proper signaling can be computed in a globally optimal manner 
(see, e.g., \cite{ZeYeGuGuZh13}).
By observing that the rate region obtained with suboptimal improper strategies can be larger than the one for globally optimal proper signaling,
it was concluded in \cite{ZeYeGuGuZh13} that proper signaling is not always the optimal strategy in the two-user interference channel
(under the assumption of Gaussian codebooks and TIN).
While this result was obtained for a single channel realization,
simulations in \cite{KiYeCh13} revealed that a similar behavior occurs for a large range of channel realizations.

However, for all these simulations, an important restriction was assumed, namely that so-called \emph{time-sharing},
i.e., averaging data rates and transmit powers over several transmit strategies (see Section~\ref{sec:TS}), is not allowed.
We refer to strategies that obey this restriction as \emph{pure strategies}.
After comparing the rate regions obtained with such pure strategies and observing gains due to improper signaling,
the authors of \cite{HoJo12} and of \cite{ZeYeGuGuZh13} took the convex hull of each rate region in order to account for the possibility of averaging the data rates.
After this operation, they still observed gains by improper signaling.
However, this does not answer the question whether improper signaling can still be necessary for optimal performance
if coded time-sharing \cite{GaKi11} is allowed, namely if both the data rates and the transmit powers can be averaged.

It is well known that coded time-sharing can in general achieve larger rate regions than a convex hull formulation (e.g., \cite{HaKo81,MoKh09}).
However, the question we consider here is a different one:
when using TIN strategies with coded time-sharing in the two-user interference channel, can improper Gaussian inputs perform better than proper Gaussian inputs?
In \cite{HeUt18}, it was observed in numerical simulations that proper signaling with (coded) time-sharing leads to larger rate regions
than the regions obtained with the convex hull operation in \cite{HoJo12,ZeYeGuGuZh13}.
However the question whether or not improper signaling with time-sharing can outperform proper signaling with time-sharing remained open.
In this paper, we settle this problem by showing that proper signaling is indeed the optimal choice if coded time-sharing is allowed---a conclusion that is completely different from the case without time-sharing.

As a further contribution, we show that symbol extensions (considering multiple subsequent channel uses as a single channel use in a higher-dimensional system, e.g., \cite{JaSh08,CaJa08,CaJaWa10})
cannot enlarge the considered time-sharing rate region, which is a generalization of a result from \cite{BeLiNaYa16} to complex scenarios.
Moreover, we discuss an algorithm to numerically compute the time-sharing rate region (Sections~\ref{sec:algo} and~\ref{sec:inner})
and revisit a numerical example from the conference contribution \cite{HeUt18} with additional interpretations (Sections~\ref{sec:num} and~\ref{sec:conclusion}).

The scenario for which we obtain these results is the two-user interference channel under the assumptions of Gaussian input signals and TIN, and there are several motivations for this combination of assumptions.
First, our aim is to clarify upon several existing publications on proper and improper signals in the two-user interference channel (e.g., \cite{HoJo12,ZeYeGuGuZh13,KiYeCh13}),
which made exactly this combination of assumptions.

Second, even though it is well known that simultaneous decoding \cite{GaKi11} can achieve better performance in general,
and even though non-Gaussian codebooks might be beneficial in the interference channel \cite{WuShVe11,WuShVe15},
TIN strategies with Gaussian inputs are among the most prominent transmit strategies studied in the literature, and it is thus interesting to understand their performance limitations.
In the case of a complex setting, this includes understanding whether or not the possibility of improper signals needs to be taken into consideration when assessing the performance limits of TIN in the two-user interference channel.

Third, proofs for the real-valued interference channel 
do not directly generalize to complex scenarios if they explicitly make use of the expression for the differential entropy of the real-valued Gaussian distribution
(e.g., the proof that using colored Gaussians over several letters is not beneficial \cite{BeLiNaYa16}
or the derivation of interference regimes in which time-sharing over TIN strategies matches the Han-Kobayashi rate region \cite{CoNa12})
since structurally equivalent entropy expressions are obtained in the complex case only under a restriction to proper signals.\footnote{The differential entropy of an improper complex random variable $\mc x$
with variance $\mc c_{\mc x}$ and pseudovariance $\mc {\rtilde c}_{\mc x}$ is given by 
$\Entr{\mc x} = \log_2 (\pi \E \mc c_{\mc x}) + \frac{1}{2}\log_2\left(1-\frac{|\mc {\rtilde c}_{\mc x}|^2}{\mc c_{\mc x}^2}\right)$ \cite{ScSc10}, i.e., it contains a second summand that is not present in the expression
for the differential entropy of a real-valued Gaussian random variable. This entails structural differences compared to the real-valued case also in the rate equations, see \eqref{eq:rk}.}
Knowing about the optimality of proper signals can enable us to transfer such results for Gaussian signals with TIN to complex scenarios by exploiting the structural similarities of the real-valued and the proper complex Gaussian entropy.

Finally, the study of TIN strategies has to be considered as a starting point, and the proof technique developed in this paper might be helpful to study scenarios with more complicated rate expressions in the future.
For instance, it is also not obvious whether or not improper Gaussian signals can improve upon proper Gaussian signals in the Han-Kobayashi coding scheme \cite{HaKo81} in the complex interference channel,
and previous work has mainly focused on studying the Han-Kobayashi scheme in real-valued scenarios.

The research presented in this paper was inspired by previous studies in the one-sided interference channel with Gaussian inputs and TIN.
For this setting, explicit characterizations of globally optimal pure strategies with improper signaling were found for the sum rate maximization problem in \cite{KuSu15}
and for the whole Pareto boundary of the rate region in \cite{LaSaSc17}.
Based on these characterizations, it was concluded that improper signaling leads to a higher sum rate and an enlarged rate region when compared to proper signaling,
and this result remained true when taking the convex hulls of the respective rate regions.
However, it was then proven analytically in \cite{HeUt17a} that improper signaling no longer brings an advantage over proper signaling if time-sharing is allowed.

The proof technique from \cite{HeUt17a} is based on the fact that the 
one-sided interference channel can be transformed to a standard form with real-valued channel coefficients,
which is not possible for the general two-user interference channel, where both users mutually disturb each other.
Therefore, the proof technique from the one-sided interference channel can unfortunately not be directly transferred.
In this paper, we overcome this problem by the novel idea of introducing an enhanced interference channel with real-valued channel coefficients (Section~\ref{sec:main}).

\emph{Notation:} 
We use $\zero$ for the zero vector, $\one$ for the all-ones vector, and $\bullet^\Tr$ for the transpose.
Inequalities for vectors have to be understood as sets of component-wise inequalities.
The matrix $\id_N$ is the $N\times N$ identity matrix. The ceiling operation $\lceil a\rceil$ rounds a real number $a$ to the next integer greater than or equal to $a$.
We use $\Re$, $\Im$, and $\angle$ for the real part, imaginary part, and the argument of a complex number, respectively.
Complex quantities are written in sans-serif font.

\section{System Model and Time-Sharing}
\label{sec:TS}
We consider a two-user interference channel 
\begin{subequations}
\label{eq:model}
\begin{align}
\label{eq:model1}
\mc y_1 &= \mc h_{11} \mc x_1 + \mc h_{12} \mc x_2 + \noise_1 \\
\label{eq:model2}
\mc y_2 &= \mc h_{21} \mc x_1 + \mc h_{22} \mc x_2 + \noise_2
\end{align}
\end{subequations}
with proper Gaussian noise $\noise_k\sim\mathcal{CN}(0,c_{\noise_k})$, where
the input signals $\mc x_k,~k=1,2$ are (possibly improper) zero-mean complex Gaussian with variance $c_{\mc x_k}=\Expect{|\mc x_k|^2}$.
The two input signals and the noise at both users (i.e., $\mc x_1$, $\mc x_2$, $\noise_1$, and $\noise_2$) are assumed to be mutually independent.
The channel coefficients and noise variances are assumed to be known and to be constant over time.

\subsection{Pure Strategies}
When applying one of the optimization methods for improper signaling from \cite{HoJo12,PaPaKiLe13,ZeYeGuGuZh13,ZeZhGuGu13,LaAgVi16,KiYeCh13,KuSu15,LaSaSc17},
the result is a single transmit strategy that is applied as long as the channel realization remains the same.
In this paper, we refer to this kind of transmit strategies as \emph{pure strategies}.

For the system under consideration, the achievable rates (Shannon rates) of the two users 
in case of a pure strategy can be expressed as (e.g., \cite{ZeYeGuGuZh13})
\begin{equation}
\label{eq:rk}
r_k(\mathcal X) = \log_2\left(\frac{c_{\mc y_k} }{c_{\mc s_k}  }\right)+\frac{1}{2}\log_2\left(
\frac{1- c_{\mc y_k}^{-2} |\mc{\rtilde c}_{\mc y_k}|^2}
{1- c_{\mc s_k}^{-2} |\mc{\rtilde c}_{\mc s_k}|^2}
\right)
\end{equation}
with 
\begin{subequations}
\begin{align}
c_{\mc y_k} &= |\mc h_{kk}|^2 c_{\mc x_k}  + c_{\mc s_k}, &
c_{\mc s_k} &=|\mc h_{kj}|^2 c_{\mc x_j}  +  c_{\noise_k}, \\
\mc{\rtilde c}_{\mc y_k} &=  \mc h_{kk}^2 \mc{\rtilde c}_{\mc x_k}  + \mc{\rtilde c}_{\mc s_k}, &
\mc{\rtilde c}_{\mc s_k} &=  \mc h_{kj}^2 \mc{\rtilde c}_{\mc x_j}
\end{align}
\end{subequations}
and $j=3-k$.
The signal $\mc s_k$, whose variance and pseudovariance is given above, can be interpreted as the interference-plus-noise at receiver $k$. 
We use $\mathcal X$ to summarize all parameters that describe the chosen strategy, i.e.,
$\mathcal X=(c_{\mc x_1},c_{\mc x_2},\mc{\rtilde c}_{\mc x_1},\mc{\rtilde c}_{\mc x_2})$ is the tuple of all transmit variances and pseudovariances.
In the special case of a strategy with proper signaling, both pseudovariances are zero (i.e., $\mc{\rtilde c}_{\mc x_1}=\mc{\rtilde c}_{\mc x_2}=0$), and the second summand in \eqref{eq:rk} vanishes.

To study the rate region that is achievable with pure strategies, we have to find Pareto-optimal pairs of achievable rates $(r_1,r_2)$.
This can be done by solving the optimization 
\begin{subequations}
\label{eq:noTS}
\begin{align}
\max_{\mathcal X,R\in\mathbb{R}} ~~R  \quad\st\quad & r_k(\mathcal X) \geq \rho_k R,~~\allk \label{eq:noTS:rate} \\
& 0\leq c_{\mc x_k}\leq P_k, ~~\allk \label{eq:noTS:pow} \\
&  |\mc{\rtilde c}_{\mc x_k}| \leq c_{\mc x_k},~~\allk \label{eq:noTS:pos}
\end{align}
\end{subequations}
where $\mb{\rho}=[\rho_1,\rho_2]^\Tr=[\beta,1-\beta]^\Tr$ for various $\beta\in[0;1]$.
This kind of optimization is called \emph{rate balancing} \cite{JoBo02}
and the vector $\mb\rho$ is sometimes referred to as \emph{rate profile} vector \cite{MoZhCi06}.
Its entries $\rho_k$ define relative rate targets of the two users, and the optimal value of $R$ is the highest possible common scaling factor that still leads to a feasible pair of rates.
If $\rho_1+\rho_2=1$, the value of $R$ equals the sum rate that is achieved by the obtained strategy.
Without loss of generality, we assume that $\mb\rho$ is chosen in this manner.

Due to the second constraint \eqref{eq:noTS:pow}, it is ensured that the average transmit power of user $k$ is nonnegative and does not exceed $P_k$.
The last constraint \eqref{eq:noTS:pos} is a requirement that has to be fulfilled by any valid combination of variance and pseudovariance (see, e.g., \cite{ScSc10}).

\subsection{Time-Sharing}
The alternative to pure strategies is that an algorithm for transceiver design delivers multiple transmit strategies
along with weighting factors $\tau_\ms$ that indicate
which fraction of the total time the $\ms$th strategy should be employed.
This concept is generally referred to as \emph{time-sharing} or \emph{coded time-sharing} (e.g., \cite{GaKi11}).
The rate balancing optimization with time-sharing can be formulated as 
\begin{subequations}
\label{eq:primal}
\begin{align}
\label{eq:primal:rate}
\max_{\substack{\mathcal X^{(\ms)},\Ms\in\mathbb{N},R\in\mathbb{R}\\\mb\tau\geq\zero: \one^\Tr\mb\tau=1}} ~~R  \quad\st\quad 
& \sum_{\ms=1}^\Ms \tau_\ms r_k(\mathcal X^{(\ms)}) \geq \rho_k R,~~\allk \\
& \sum_{\ms=1}^\Ms \tau_\ms c_{\mc x_k}^{(\ms)}\leq P_k, ~~\allk \label{eq:primal:pow}\\
& 0\leq c_{\mc x_k}^{(\ms)}, ~~\allk,~\forall\ms \label{eq:primal:pos}\\
&  |\mc{\rtilde c}_{\mc x_k}^{(\ms)}| \leq c_{\mc x_k}^{(\ms)},~~\allk, ~\forall \ms 
\label{eq:primal:pcov}
\end{align}
\end{subequations}
where $\mathcal X^{(\ms)}=(c_{\mc x_1}^{(\ms)},c_{\mc x_2}^{(\ms)},\mc{\rtilde c}_{\mc x_1}^{(\ms)},\mc{\rtilde c}_{\mc x_2}^{(\ms)})$
are the transmit (pseudo)variances employed in the $\ms$th strategy
and $\mb{\tau}=[\tau_1, \dots, \tau_\Ms]$ is the vector of time-sharing weights.
If we interpret time-sharing as subsequent application of $\Ms$ strategies, these weights correspond to the lengths of the time intervals.
In an alternative interpretation from \cite{HaKo81},
a time-sharing parameter $Q\in\{1,\dots,\Ms\}$ randomly decides which strategy is employed \cite{HaKo81},
and $\tau_\ms$ specifies the probability that the strategy $Q=\ms$ is chosen.

\subsection{Convex Hull}
Many researchers have accounted for the possibility of time-sharing by 
first deriving a method for optimizing pure strategies and then
taking the convex hull of the rate region obtained with pure strategies (e.g., \cite{HoJo12,ZeYeGuGuZh13,ZeZhGuGu13,LaSaSc17}).
However, it was pointed out in \cite{HaKo81} that this does not exploit the full potential of time-sharing.

The reason for this is that taking the convex hull of the rate region can be interpreted as averaging the achievable data rates over several operation points
while respecting the power constraints individually in each operation point, i.e., 
\begin{subequations}
\label{eq:rts}
\begin{align}
\label{eq:rts:rate}
\max_{\substack{\mathcal X^{(\ms)},\Ms\in\mathbb{N},R\in\mathbb{R}\\\mb\tau\geq\zero: \one^\Tr\mb\tau=1}} ~~R  \quad\st\quad 
& \sum_{\ms=1}^\Ms \tau_\ms r_k(\mathcal X^{(\ms)}) \geq \rho_k R,~~\allk \\
& c_{\mc x_k}^{(\ms)}\leq P_k, ~~\allk,~\forall\ms \label{eq:rts:pow}\\
& 0\leq c_{\mc x_k}^{(\ms)}, ~~\allk,~\forall\ms \label{eq:rts:pos}\\
&  |\mc{\rtilde c}_{\mc x_k}^{(\ms)}| \leq c_{\mc x_k}^{(\ms)},~~\allk, ~\forall \ms.
\label{eq:rts_pcov}
\end{align}
\end{subequations}
However, for coded time-sharing as discussed, e.g., in \cite{HaKo81}, these constraints are only required to be fulfilled on average.
Consequently, the convex hull formulation is more restrictive,
which can be seen by comparing \eqref{eq:primal:pow} to \eqref{eq:rts:pow},
and might thus not achieve the complete time-sharing rate region.

\subsection{Comparison}
If the aim is to optimize a time-sharing strategy, we cannot first optimize pure strategies for fixed transmit power limitations and account for time-sharing afterwards.
Due to the possibility of averaging the transmit powers, it is not clear a~priori under which transmit power constraints the pure strategies should be optimized.
Instead, we have to account for the possibility of time-sharing already in the optimization procedure.

At first glance, one might get the impression that time-sharing leads to undesirable fluctuations of the transmit powers since
the power constraints are fulfilled only on average.
To understand whether or not this is an issue, we need to recall that even the power constraint \eqref{eq:noTS:pow}
in the case of pure strategies is an average power constraint
since it restricts the expected value of the squared transmit signal and not the peak value.
As Gaussian codebooks are assumed in all the abovementioned studies based on Shannon rates, the fluctuations of the instantaneous transmit powers can thus be significant
even in case of pure strategies.
This result carries over to the case where the convex hull is taken.\footnote{Accordingly, 
the convex hull formulation was called \emph{time-sharing under short-term average power constraints} in \cite{HeUt17a} while coded time-sharing was called \emph{time-sharing under long-term average power constraints}.}

The relaxed power constraints in the time-sharing formulation are indeed an additional source for fluctuations of the transmit powers.
However, there is no need that the various strategies $\ms=1,\dots,\Ms$ are applied one after another.
In the coded time-sharing formulation based on a time-sharing parameter $Q$, 
it is randomly decided on a per-symbol basis which of the strategies is applied (see \cite{HaKo81}).
From a technical perspective, this could be approximated by interleaving 
transmit symbols belonging to different strategies in a fixed pseudorandom ordering.
If we take this perspective on time-sharing, the fluctuations of the transmit powers are short-term fluctuations,
just like in the case of pure strategies or in case of the convex hull formulation.

\subsection{Remark on Symbol Extensions}
All rate expressions given above can be extended to include the possibility of symbol extensions (see, e.g., \cite{JaSh08,CaJa08,CaJaWa10,BeLiNaYa16,NaNg19}).
However, we have omitted this possibility since we show below that symbol extensions do not bring any advantage in the considered scenario.
Indeed, several publications have previously studied potential gains by symbol extensions in the real-valued interference channel \cite{ChVe93,BeLiNaYa16,NaNg19},
and it was shown that symbol extensions do not bring an advantage in the case of Gaussian signals and TIN with coded time-sharing \cite{BeLiNaYa16}.
However, these existing publications have not studied the complex interference channel, where the possibilities of symbol extensions and improper signals need to be considered jointly.
Thus, one part of proving Theorem~\ref{th:main} in the next section is to extend this result to complex settings.
For the sake of a clear presentation of the main ideas, the main part of the paper gives intuitive justifications without considering symbol extensions
while formal proofs including the aspect of symbol extensions are deferred to the appendix.

\section{Main Result}
\label{sec:main}
For the rate region with pure strategies as well as for its convex hull,
is has been observed that improper signaling can lead to a larger region than proper signaling (see the summary in Section~\ref{sec:intro}).
The following theorem shows that this result changes when considering the rate region with coded time-sharing.

\begin{theorem}
\label{th:main}
Consider the two-user Gaussian interference channel \eqref{eq:model}
with Gaussian input signals under power constraints \eqref{eq:primal:pow}, and assume that interference is treated as noise.
Then, the whole time-sharing rate region $\mathcal{R}$ can be achieved using proper input signals without symbol extensions.
\end{theorem}

The proof is established by combining three Lemmas that are stated and proven below.
Lemma~\ref{lem:enhanced} introduces an enhanced interference channel whose rate region $\bar{\mathcal{R}}$ contains the whole original rate region $\mathcal{R}$.
Lemma~\ref{lem:proper_equal} shows that the rate regions of both systems under a restriction to proper signaling without symbol extensions, i.e., $\mathcal{R}_\mt{proper}$ and $\bar{\mathcal{R}}_\mt{proper}$, coincide.
Finally, Lemma~\ref{lem:proper_opt} shows that proper signaling without symbol extensions achieves the whole rate region of the enhanced system, 
i.e., $\bar{\mathcal{R}}_\mt{proper}$ is the same as $\bar{\mathcal{R}}$.
\begin{proof}[Proof of Theorem~\ref{th:main}]
By Lemmas~\ref{lem:enhanced}, \ref{lem:proper_equal}, and~\ref{lem:proper_opt}, we have
$\mathcal{R}\subseteq\bar{\mathcal{R}}=\bar{\mathcal{R}}_\mt{proper}=\mathcal{R}_\mt{proper}$,
where $\bar{\mathcal{R}}$ is defined in Lemma~\ref{lem:enhanced}, and 
the subscript ${}_\mt{proper}$ denotes the respective rate region under a restriction to proper input signals without symbol extensions.
On the other hand, it is clear that $\mathcal{R}_\mt{proper}\subseteq\mathcal{R}$.
This shows that $\mathcal{R}_\mt{proper}=\mathcal{R}$.
\end{proof}

\begin{lemma}
\label{lem:enhanced}
Let $\bar{\mathcal R}$ denote the time-sharing rate region of the modified interference channel
\begin{subequations}
\label{eq:modelUB}
\begin{align}
\label{eq:modelUB1}
\mc y_1 &= |\mc h_{11}| \mc x_1 + |\mc h_{12}| \mc x_2 + \noise_1 \\
\label{eq:modelUB2}
\mc y_2 &= |\mc h_{21}| \mc x_1 + |\mc h_{22}| \mc x_2 + \noise_2
\end{align}
\end{subequations}
under the same assumptions as in Theorem~\ref{th:main}.
Then, $\mathcal{R}\subseteq\bar{\mathcal R}$.
\end{lemma}

The following intuitive justification focuses on the main novelty in the proof, namely on how to deal with the different phases in a complex setting.
The formal proof of Lemma~\ref{lem:enhanced} including the possibility of symbol extensions is presented separately in the appendix.

In the original interference channel \eqref{eq:model},
let\footnote{The time slot index $\ms$ can be omitted for the sake of brevity whenever we consider only a particular time slot.} $\mc{\rtilde c}_{\mc x_k} = \kappa_{\mc x_k} \E^{\J\varphi_k},~\allk$ with the nonnegative impropriety coefficient $\kappa_{\mc x_k}\geq0$.
We note that $|\mc{\rtilde c}_{\mc s_k}|^2=|\mc h_{kj}|^4 \kappa_{\mc x_j}^2$,
so that the only dependence of $r_k(\mathcal X)$ from \eqref{eq:rk} on $\varphi_k$, $\varphi_j$ and on the phases of the channel coefficients is via $\mc{\rtilde c}_{\mc y_k}$.
Moreover, we have
\begin{align}
|\mc{\rtilde c}_{\mc y_k}|^2 &= |\mc h_{kk}^2 \mc{\rtilde c}_{\mc x_k} +\mc h_{kj}^2 \mc{\rtilde c}_{\mc x_j}|^2 \nonumber
\\&\geq \left| |\mc h_{kk}^2 \mc{\rtilde c}_{\mc x_k}| - |\mc h_{kj}^2 \mc{\rtilde c}_{\mc x_j}| \right|^2 \nonumber
\\&=\left| |\mc h_{kk}|^2 \kappa_{\mc x_k} - |\mc h_{kj}|^2 \kappa_{\mc x_j} \right|^2
\end{align}
with equality if 
\begin{align}
\angle(\mc h_{kk}^2 \mc{\rtilde c}_{\mc x_k}) &= \pi+\angle(\mc h_{kj}^2 \mc{\rtilde c}_{\mc x_j})
\ifCLASSOPTIONdraftcls\else\nonumber\\\fi
~\Leftrightarrow~
\varphi_k \ifCLASSOPTIONdraftcls\else\tmpand\fi= \pi + \varphi_j+ \angle(\mc h_{kj}^2\mc h_{kk}^{-2}).
\end{align}
An upper bound to the rate $r_k(\mathcal X)$ is thus given by
\begin{align}
\label{eq:rkUB}
&{\bar r}_k(\mathcal X) = \log_2\left(\frac{c_{\mc y_k} }{c_{\mc s_k}  }\right)
\ifCLASSOPTIONdraftcls\else\nonumber\\\tmpand\fi
+\frac{1}{2}\log_2\left(
\frac{1- c_{\mc y_k}^{-2} \left| |\mc h_{kk}|^2 \kappa_{\mc x_k} - |\mc h_{kj}|^2 \kappa_{\mc x_j} \right|^2}
{1- c_{\mc s_k}^{-2} |\mc h_{kj}|^4 \kappa_{\mc x_j}^2}
\right).
\end{align}
This upper bound neither depends on $\varphi_k,\varphi_j$ nor on the phases of the channel coefficients.
To see that it is indeed an upper bound, note that the rate expression \eqref{eq:rk} is non-increasing in $|\mc{\rtilde c}_{\mc y_k}|^2$
since the denominator of the second summand in \eqref{eq:rk} is positive due to $c_{\mc s_k}^2 > |\mc{\rtilde c}_{\mc s_k}|^2$ (see \cite[Lemma~3]{ZeYeGuGuZh13}).

As the upper bound ${\bar r}_k(\mathcal X)$ does not depend on the phases of the channel coefficients, it is equal for the original interference channel \eqref{eq:model}
and for the enhanced interference channel \eqref{eq:modelUB}.
A choice of $\varphi_1$ and $\varphi_2$ that achieves equality in \eqref{eq:rkUB} for both users simultaneously exists
if $ \angle(\mc h_{12}\mc h_{11}^{-1}) = -\angle(\mc h_{21}\mc h_{22}^{-1})$.
This condition is fulfilled in \eqref{eq:modelUB}, but not necessarily in \eqref{eq:model}.
Thus, for any time-sharing solution with $\varphi_1$ and $\varphi_2$ being chosen optimally in each strategy,
the average rates achieved in the enhanced interference channel are at least as high as in the original system.
This is the statement of Lemma~\ref{lem:enhanced}.

The idea of channel enhancement has previously been used, e.g., to study the MIMO broadcast channel \cite{WeStSh06}, the MIMO wiretap channel \cite{LiSh09}, and the MIMO relay channel \cite{GeHeWeUt15}.
In all these cases, the authors increased channel gains (or, equivalently, reduced noise) in a way that the resulting scenarios became degraded.
The enhanced channel constructed here is different in several respects. First of all, the objective is not making a MIMO scenario degraded.
Instead a single-input single-output (SISO) scenario is modified in a way that we obtain a standard form with real-valued channels (which is otherwise impossible except in special cases such as the one-sided interference channel).
Moreover, the enhancement is performed without changing the magnitude of any channel coefficient, but only by adapting the phases,
and it is remarkable that the obtained system can still be proven to be always superior or equal to the original one.

\begin{lemma}
\label{lem:proper_equal}
Assume a constraint that all transmit signals have to be proper without symbol extensions, and 
let $\mathcal{R}_\mt{proper}$ and $\bar{\mathcal{R}}_\mt{proper}$ denote the resulting time-sharing rate regions
of \eqref{eq:model} and \eqref{eq:modelUB}, respectively.
Then, under the assumptions of Theorem~\ref{th:main}, $\mathcal{R}_\mt{proper}=\bar{\mathcal{R}}_\mt{proper}$.
\end{lemma}
\begin{proof}
If $\mc{\rtilde c}_{\mc x_1}=\mc{\rtilde c}_{\mc x_2}=0$, $r_k(\mathcal X)$ does not depend on the phases of the channel coefficients.
\end{proof}

\begin{lemma}
\label{lem:proper_opt}
For the enhanced interference channel \eqref{eq:modelUB} under the assumptions of Theorem~\ref{th:main},
proper signaling without symbol extensions achieves the whole time-sharing rate region, i.e., 
$\bar{\mathcal{R}}_\mt{proper}=\bar{\mathcal{R}}$.
\end{lemma}
We again give an intuitive justification without considering the possibility of symbol extensions and defer the formal proof to the appendix.
We switch to the composite real representation, where complex vectors $\mbc b$ and linear operations $\mbc b \mapsto \mbc A \mbc b$ (with a complex matrix $\mbc A$) are represented by
\begin{align}
\crr b&= \begin{bmatrix}
\Re\mbc b\\\Im\mbc b
\end{bmatrix}
&
\label{eq:crrmat}
\crr b &\mapsto \begin{bmatrix}
\Re\mbc A & -\Im\mbc A\\\Im\mbc A& \Re\mbc A
\end{bmatrix}\crr b .
\end{align}
For further details, see, e.g., \cite{AdScSc11}.
If $\mbc b$ is a random vector with covariance matrix $\mbc C_{\mbc b}$ and pseudocovariance matrix $\mbc{\rtilde C}_{\mbc b}$, its composite real covariance matrix (i.e., the covariance matrix of $\crr b$) is given by
(e.g., \cite{HeUt15})
\begin{align}
\mb C_{\crr b} &= \frac{1}{2}\left(\begin{bmatrix}
\Re \mbc{C}_{\mbc b} & - \Im \mbc{C}_{\mbc b} \\
\Im \mbc{C}_{\mbc b} & \Re \mbc{C}_{\mbc b}
\end{bmatrix}+\begin{bmatrix}
\Re \mbc{\rtilde C}_{\mbc b} & \Im \mbc{\rtilde C}_{\mbc b} \\
\Im \mbc{\rtilde C}_{\mbc b} & - \Re \mbc{\rtilde C}_{\mbc b}
\end{bmatrix}\right).
\label{eq:crrcov}
\end{align}

Since the enhanced interference channel \eqref{eq:modelUB} has real-valued channel coefficients, 
its composite real representation based on \eqref{eq:crrmat} reads as
\begin{subequations}
\begin{align}
\crr y_1 &= \begin{bmatrix}|\mc h_{11}| & 0 \\ 0 & |\mc h_{11}|\end{bmatrix} \crr x_1 + \begin{bmatrix}|\mc h_{12}| & 0 \\ 0 & |\mc h_{12}|\end{bmatrix} \crr x_2 + \crr\eta_1\\
\crr y_2 &= \begin{bmatrix}|\mc h_{21}| & 0 \\ 0 & |\mc h_{21}|\end{bmatrix} \crr x_1 + \begin{bmatrix}|\mc h_{22}| & 0 \\ 0 & |\mc h_{22}|\end{bmatrix} \crr x_2 + \crr\eta_2
\end{align}
\end{subequations}
with real-valued Gaussian noise $\crr\eta_k\sim\mathcal{N}(\zero,\frac{c_{\noise_k}}{2}\id_2),~\forall k$.
This description is mathematically equivalent to a symbol extension over two symbols in a real-valued system with constant channels.

For such a real-valued setting, it was shown in \cite[Th.~2]{BeLiNaYa16} that diagonal covariance matrices are optimal.
As \eqref{eq:crrcov} yields
\begin{align}
\mb C_{\crr x_k}^{(\ms)} &= \frac{1}{2}\begin{bmatrix}
c_{\mc x_k}^{(\ms)} + \Re \mc{\rtilde c}_{\mc x_k}^{(\ms)} & \Im \mc{\rtilde c}_{\mc x_k}^{(\ms)} \\
\Im \mc{\rtilde c}_{\mc x_k}^{(\ms)} & c_{\mc x_k}^{(\ms)} - \Re \mc{\rtilde c}_{\mc x_k}^{(\ms)}
\end{bmatrix}
\end{align}
this means that we can directly set\footnote{\label{foot:pcov_real}An alternative derivation is as follows. In the enhanced interference channel \eqref{eq:modelUB}, it is optimal to choose 
$\varphi_1=\pi+\varphi_2$
so that the upper bound \eqref{eq:rkUB} is achieved.
We can thus use ${\bar r}_k$ from \eqref{eq:rkUB} as rate expression in the enhanced system.
As ${\bar r}_k$ does not depend on $\varphi_1$ and $\varphi_2$,
we may choose $\varphi_2=0$ w.l.o.g., so that $\mc{\rtilde c}_{\mc x_1}=-\kappa_{\mc x_1}\leq 0$ and $\mc{\rtilde c}_{\mc x_1}=\kappa_{\mc x_2}\geq 0$ are both real-valued.} $\Im \mc{\rtilde c}_{\mc x_k}^{(\ms)}=0$, so that
\begin{align}
\mb C_{\crr x_k}^{(\ms)} &= \frac{1}{2}\begin{bmatrix}
c_{\mc x_k}^{(\ms)} + \mc{\rtilde c}_{\mc x_k}^{(\ms)} & 0 \\
0 & c_{\mc x_k}^{(\ms)} - \mc{\rtilde c}_{\mc x_k}^{(\ms)}
\end{bmatrix}=: \begin{bmatrix}
\pow_{k,1}^{(\ms)} & 0 \\
0 & \pow_{k,2}^{(\ms)}
\end{bmatrix}
\label{eq:per_carrier_powers}
\end{align}
Moreover, as remarked in \cite[Sec.~III]{BeLiNaYa16}, any rate achievable with the covariance matrices \eqref{eq:per_carrier_powers}
can also be achieved by time-sharing over strategies without symbol extensions (equivalent to $\pow_{k,1}^{(\ms)} = \pow_{k,2}^{(\ms)}$).
To see why this is true, we can use a similar argument as in \cite{NaNg19} and use $L'=2L$ time slots with $\tau_{\ms}'=\tau_{\lceil \ms/2 \rceil} / 2$ in \eqref{eq:primal},
and we can then set
\begin{align}
\pow_{k,1}^{\prime(\ms)} = \pow_{k,2}^{\prime(\ms)} =\begin{cases}
\pow_{k,1}^{(\lceil\ms/2\rceil)},\quad \text{$\ms$ odd}, \\
\pow_{k,2}^{(\lceil\ms/2\rceil)},\quad \text{$\ms$ even}.
\end{cases}
\end{align}
This does not change the value on the left hand side of \eqref{eq:primal:pow},
and since the diagonal covariance matrices $\mb C_{\crr x_k}^{(\ms)}$ lead to
\begin{align}
{\bar r}_{k}^{(\ms)} &= \sum_{n=1}^2 \frac{1}{2}\log_2\left(1 + \frac{|\mc h_{kk}|^2 \pow_{k,n}^{(\ms)}}{\frac{c_{\noise_k}}{2}+|\mc h_{kj}|^2 \pow_{j,n}^{(\ms)}}\right)
\end{align}
the value on the left hand side of \eqref{eq:primal:rate} remains unchanged as well.

Translating this back to the complex representation by means of \eqref{eq:per_carrier_powers}, we obtain a strategy which, for all time slots $\ms$, fulfills
\begin{equation}
c_{\mc x_k}^{\prime(\ms)} + \mc{\rtilde c}_{\mc x_k}^{\prime(\ms)} = c_{\mc x_k}^{\prime(\ms)} - \mc{\rtilde c}_{\mc x_k}^{\prime(\ms)} ~~\Leftrightarrow~~
\mc{\rtilde c}_{\mc x_k}^{\prime(\ms)}=0.
\end{equation}
Therefore, there always exists a solution to \eqref{eq:primal} with vanishing pseudovariances in all time slots, i.e., with proper signaling.
By extending this argumentation to also consider the possibility of symbol extensions in the complex setting (see the formal proof in the appendix), we obtain the statement of Lemma~\ref{lem:proper_opt}, which completes the proof of the main result.

\section{Algorithmic Solution}
\label{sec:algo}
Having established that proper Gaussian signals are the optimal Gaussian signals in the two-user interference channel with TIN,
we are interested in calculating the corresponding achievable rate region with coded time-sharing.
To this end, we consider the Lagrangian dual problem (e.g., \cite{BoVa09,BaShSh06}) of \eqref{eq:primal}.
This is a valid approach since
it is easy to verify that the so-called time-sharing condition from \cite{YuLu06} is fulfilled for this problem, which implies that its duality gap vanishes even though the rate expressions are nonconcave \cite{YuLu06}.
This zero-duality-gap property is also confirmed in Section~\ref{sec:recov}, where we recover a solution to the primal problem.

Let $\mb\pow^{(\ms)}=[\pow_1^{(\ms)},\pow_2^{(\ms)}]^\Tr$, and define the rate with proper signals as $r_k(\mb\pow):=\left.r_k(\mathcal{X})\right|_{\mathcal{X} = (\pow_1,\pow_2,0,0)}$ with $r_k(\mathcal{X})$ from \eqref{eq:rk}.
For the case of proper signals, the constraint \eqref{eq:primal:pcov} can be dropped.
We introduce the dual variables $\mb\mu=[\mu_1,\mu_2]^\Tr$ and $\mb\lambda=[\lambda_1,\lambda_2]^\Tr$,
and we dualize the constraints \eqref{eq:primal:rate}--\eqref{eq:primal:pow} to obtain the dual problem
\begin{align}
\label{eq:dual_start}
\min_{\substack{\mb\mu\geq\zero\\\mb\lambda\geq\zero}} ~\max_{\substack{\Ms\in\mathbb{N},R\in\mathbb{R}\\(\mb\tau\geq\zero): \one^\Tr\mb\tau=1}} ~\max_{(\mb\pow^{(\ms)}\geq\zero)_{\forall\ms}} ~~ \Lag
\end{align}
with the Lagrangian function
\begin{mueq}
\label{eq:lagrangian}
\Lag = R \,+ \sum_{k=1}^2 \Bigg(\mu_k \Bigg(\sum_{\ms=1}^\Ms \tau_\ms \, r_k(\mb\pow^{(\ms)}) -\rho_k R\Bigg)
\ifCLASSOPTIONdraftcls\else\\\fi
+ \lambda_k\Bigg(P_k-\sum_{\ms=1}^\Ms \tau_\ms \, \pow_k^{(\ms)}\Bigg) \Bigg).
\end{mueq}
The reformulation 
\begin{mueq}
\label{eq:lagrangian_reform}
\Lag = \left(1- \sum_{k=1}^2 \mu_k\rho_k \right)R ~+~ \sum_{k=1}^2 \lambda_k P_k 
\ifCLASSOPTIONdraftcls\else\\\fi
+\,\sum_{\ms=1}^\Ms \tau_\ms \sum_{k=1}^2\left( \mu_k \, r_k(\mb\pow^{(\ms)}) - \lambda_k \pow_k^{(\ms)} \right) 
\end{mueq}
reveals that the outer minimization must choose $\mb\mu$ in a way that $\mb\rho^\Tr\mb\mu= 1$.
Otherwise, the maximization over $R$ would be unbounded, which would clearly not be optimal in terms of the minimization over $\mb\mu$.

For the inner maximization over $\mb\pow^{(\ms)}$, we have to solve
\begin{equation}
\label{eq:inner_l}
\max_{\mb\pow^{(\ms)}\geq\zero} ~ f_{\mb\mu,\mb\lambda}(\mb\pow^{(\ms)})
\end{equation}
with
\begin{equation}
\label{eq:inner:func}
f_{\mb\mu,\mb\lambda}(\mb\pow)= \sum_{k=1}^2\left( \mu_k \, r_k(\mb\pow) - \lambda_k \pow_k \right).
\end{equation}
Since the objective function and the constraint set are the same for all $\ms$,
there exists a solution in which the optimizer of \eqref{eq:inner_l} is the same for all $\ms$.
We can thus write
\begin{equation}
\label{eq:inner}
\mb\pow^\star(\mb\mu,\mb\lambda) = 
\argmax_{\mb\pow\geq\zero} ~ f_{\mb\mu,\mb\lambda}(\mb\pow)
\end{equation}
without a dependence on $\ms$.

Consequently, the dual problem simplifies to
\begin{equation}
\label{eq:dual_simpl}
\min_{\substack{\mb\mu\geq\zero,\mb\lambda\geq\zero\\\mb\rho^\Tr\mb\mu= 1}}
~\max_{\substack{\Ms\in\mathbb{N}\\(\mb \tau\geq\zero): \one^\Tr\mb\tau=1}} ~
\sum_{k=1}^2 \lambda_k P_k
+ f_{\mb\mu,\mb\lambda}(\mb\pow^\star(\mb\mu,\mb\lambda))\sum_{\ms=1}^\Ms \tau_\ms.
\end{equation}
Since $\sum_{\ms=1}^\Ms \tau_\ms=1$ is a constant, the maximum operator can be dropped, i.e., we have to solve
\begin{equation}
\label{eq:dual_final}
\min_{\substack{\mb\mu\geq\zero,\mb\lambda\geq\zero\\\mb\rho^\Tr\mb\mu= 1}} ~~
\sum_{k=1}^2 \lambda_k P_k
+ f_{\mb\mu,\mb\lambda}(\mb\pow^\star(\mb\mu,\mb\lambda)).
\end{equation}

In the following subsections, we first discuss how this outer minimization can be solved
and how an optimal solution of the primal problem \eqref{eq:primal} can be reconstructed from the dual solution.
Afterwards, we discuss a method to solve the inner problem \eqref{eq:inner}.

\subsection{Outer Problem}
To solve the outer minimization, we can apply the cutting plane method \cite{Ke60,BaShSh06},
which successively refines a lower bound that is obtained by a relaxation of \eqref{eq:dual_final}.
To this end, we first introduce a slack variable $z$ and rewrite the problem as 
\begin{mueq}
\label{eq:dual_CP_prelim}
\min_{\substack{\mb\mu\geq\zero,\mb\lambda\geq\zero,z\in\mathbb{R}\\\mb\rho^\Tr\mb\mu= 1}} ~~
 z
\ifCLASSOPTIONdraftcls\quad\else\\\fi
\st~~ z \geq 
\sum_{k=1}^2 \lambda_k P_k +
f_{\mb\mu,\mb\lambda}(\mb\pow) ~~\forall \mb\pow\geq\zero.
\end{mueq}
This is equivalent to \eqref{eq:dual_final} since the maximizer $\mb\pow^\star(\mb\mu,\mb\lambda)$ in \eqref{eq:dual_final} corresponds to the value of $\mb\pow$ that leads to the strictest inequality in \eqref{eq:dual_CP_prelim}
due to $f_{\mb\mu,\mb\lambda}(\mb\pow^\star(\mb\mu,\mb\lambda))\geq f_{\mb\mu,\mb\lambda}(\mb\pow)$ for all $\mb\pow\geq\zero$.

A relaxed version of the problem can now be obtained by replacing the uncountable constraints on $z$ by a finite set of constraints, i.e.,
\begin{mueq}
\label{eq:dual_CP}
\min_{\substack{\mb\mu\geq\zero,\mb\lambda\geq\zero,z\in\mathbb{R}\\\mb\rho^\Tr\mb\mu= 1}} ~~
 z
\ifCLASSOPTIONdraftcls\quad\else\\\fi
\st~~ z \geq 
\sum_{k=1}^2 \lambda_k P_k +
f_{\mb\mu,\mb\lambda}(\mb\pow^{(\ms)}) ~~\forall {\ms\in\{1,\dots,\Ms\}}
\end{mueq}
with given constants 
$\mb\pow^{(\ms)}$ for $\ms\in\{1,\dots,\Ms\}$.
As a consequence,
$f_{\mb\mu,\mb\lambda}(\mb\pow^{(\ms)})$ is a linear function of $\mb\mu$ and $\mb\lambda$,
and the relaxed problem \eqref{eq:dual_CP} is a linear program, for which efficient standard solvers can be used.

By solving the relaxed problem, we obtain a lower bound to the optimal value of \eqref{eq:dual_final}.
This bound can be refined by increasing $\Ms$, i.e., by adding further constants $\mb\pow^{(\ms)}$,
and it gets eventually tight if $\mb\pow^{(\ms)}=\mb\pow^\star(\mb\mu^\star,\mb\lambda^\star)$ for some $\ms$, where $(\mb\mu^\star,\mb\lambda^\star)$ is the optimizer of \eqref{eq:dual_final}.

The cutting plane method \cite{Ke60,BaShSh06} summarized in Algorithm~\ref{algo:cp} is based on this idea of successive refinement of the lower bound.
Convergence of the generated sequence $z^{(\Ms)}$, $\Ms=2,3,\dots$ to the optimal value of \eqref{eq:dual_final}
can be concluded from the convergence proof in \cite{Ke60}.
To check for convergence, we can use an upper bound to the optimal value of \eqref{eq:dual_final}
that is obtained by setting $\mb\pow^{(\ms)}=\mb\pow^\star(\mb\mu^{(\ms)},\mb\lambda^{(\ms)})$
for some $(\mb\mu^{(\ms)},\mb\lambda^{(\ms)})$
and by calculating the achievable value 
\begin{equation}
\Psi_\ms = \sum_{k=1}^2\lambda^{(\ms)}_k P_k + f_{\mb\mu^{(\ms)},\mb\lambda^{(\ms)}}(\mb\pow^{(\ms)}).
\end{equation}
To initialize the cutting plane method, we can set $\mb\pow^{(1)}$ to an arbitrary feasible transmit strategy, e.g., $\pow_k^{(1)}=\frac{1}{2}P_k$ for all $k$.
\begin{algorithm}[h]
\caption{Cutting Plane Method for Problem~\eqref{eq:dual_final}}
\label{algo:cp}
Given $L=1$ and an initialization 
$\mb\pow^{(1)}$:
\begin{enumerate}
\item Solve the linear program \eqref{eq:dual_CP} and store the optimizer in $(\mb\mu^{(\Ms+1)},\mb\lambda^{(\Ms+1)},z^{(\Ms+1)})$.\label{item:cp_lp}
\item Solve \eqref{eq:inner} to obtain $\mb\pow^{(\Ms+1)}\gets\mb\pow^\star(\mb\mu^{(\Ms+1)},\mb\lambda^{(\Ms+1)})$.\label{item:cp_inner}
\item Set $\Ms\gets \Ms+1$, and repeat Steps~\ref{item:cp_lp}) and~\ref{item:cp_inner}) until
$
\min_{\ms\in\{2,\dots,\Ms\}} \Psi_\ms 
-
z^{(\Ms)}
\leq\cpeps$.
\end{enumerate}
\end{algorithm}

\subsection{Primal Recovery}
\label{sec:recov}
As described in \cite{BaShSh06}, a possible method to recover a primal solution, i.e., a solution to the original problem \eqref{eq:primal},
is to consider the dual linear program of the cutting plane problem \eqref{eq:dual_CP}.
This problem reads as
\begin{subequations}
\begin{align}
\label{eq:recov_first}
\max_{\mb \tau\geq\zero,R\in\mathbb{R}} \!\!\!\!&~~~~
\min_{\mb\mu\geq\zero,\mb\lambda\geq\zero,z\in\mathbb{R}} ~~
\\ R&+z(1-\one^\Tr\mb\tau)+\sum_{k=1}^2 \lambda_k \sum_{\ms=1}^{\Ms} \tau_\ms\left( P_k     - \pow_k^{(\ms)} \right)
\\ &+\sum_{k=1}^2 \mu_k     \left(-R\rho_k + \sum_{\ms=1}^{\Ms} \tau_\ms\, r_k(\mb\pow^{(\ms)})\right)
\end{align}
\end{subequations}
where we have introduced the dual variable $R$ for the constraint $\mb\rho^\Tr\mb\mu=1$ and the dual variables $\mb\tau=[\tau_1,\dots,\tau_\Ms]^\Tr$ for the constraints on $z$.
Since $\mb\tau$ and $R$ have to be chosen in the outer maximization in a way that avoids that the inner minimization is unbounded,
we obtain the following reformulation with three new constraints:
\begin{subequations}
\label{eq:recov}
\begin{align}
\max_{\substack{\mb \tau\geq\zero,R\in\mathbb{R}\\\one^\Tr\mb\tau=1}}
~~R \quad\st\quad  &\sum_{\ms=1}^\Ms\tau_\ms \, r_k(\mb\pow^{(\ms)})  \geq \rho_k R,~~\allk
\\&\sum_{\ms=1}^\Ms\tau_\ms \, \pow_k^{(\ms)} \leq P_k,~~\allk.
\end{align}
\end{subequations}

As linear programs have zero duality gap, \eqref{eq:recov} has the same optimal value as \eqref{eq:dual_CP}, which converges to the optimal value of \eqref{eq:dual_final}.
This optimum is an upper bound to the solution of the original primal problem \eqref{eq:primal}
since \eqref{eq:dual_final} is the Lagrangian dual problem of \eqref{eq:primal} (weak duality, e.g., \cite{BaShSh06}).
However, since any solution of \eqref{eq:recov} clearly corresponds to a feasible strategy in \eqref{eq:primal}, this value is at the same time a lower bound to the solution
of \eqref{eq:primal}.
This shows that strong duality holds for \eqref{eq:primal}, i.e., the duality gap is zero,
which is in line with the general considerations about the optimization of time-sharing strategies in \cite{YuLu06}.

The number of strategies $\Ms$ is obtained from the execution of the cutting plane algorithm.
In principle, this number can be arbitrarily high, but 
usually, only a small number of strategies obtain nonzero time-sharing weights $\tau_\ms$ when solving \eqref{eq:recov}.
As time-sharing can be interpreted as a convex hull operation on a connected set in a rate-power-space with $4$ dimensions, 
an extension to the Carath\'eodory Theorem discussed in \cite{HaRa51} implies that
there always exists an optimal solution of \eqref{eq:primal} that requires no more than $4$ active strategies.

\section{Algorithmic Solution to the Inner Problem}
\label{sec:inner}
The remaining missing element to compute a solution to problem \eqref{eq:primal} is a solver for the inner problem \eqref{eq:inner}.
In this section, we propose to apply the branch-and-bound algorithm \cite[Sec.~6.2]{Tu16} to obtain an \mbox{$\brbeps$-optimal} solution, i.e., 
a solution that is no more than $\brbeps$ away from the global optimum of \eqref{eq:inner}.

\subsection{Monotonicity Bounds}
By plugging in the expression for $r_k(\mb\pow)$ into \eqref{eq:inner:func}, we obtain
\begin{equation}
\brborigfun(\brbp)= \sum_{k=1}^2\left( \mu_k
\log_2\left(1 + \frac{|\mc h_{kk}|^2 \brbpk}{c_{\noise_k}+|\mc h_{kj}| \brbpk[j]}\right)
 - \lambda_k \brbpk \right)
\end{equation}
where we have used $\brborigfun$ as an abbreviation for $f_{\mb\mu,\mb\lambda}$.
We introduce the extended function
\begin{mueq}
\brbfun(\brbx,\brby)=
\ifCLASSOPTIONdraftcls\else\\\fi
 \sum_{k=1}^2\left( \mu_k  \log_2\left(1+\frac{|\mc h_{kk}|^2 \brbxk }{c_{\noise_k}+|\mc h_{kj}|^2 \brbyk[j]   }\right)   - \lambda_k \brbyk \right)
\end{mueq}
which is obviously nondecreasing in $\brbx =[\brbxk[1],\brbxk[2]]^\Tr\geq\zero$ and nonincreasing in $\brby=[\brbyk[1],\brbyk[2]]^\Tr\geq\zero$.
In the following, we establish an upper and a lower bound that are based on these monotonicity properties.

The inequality
\begin{align}
\label{eq:utop}
\brborigfun(\brbp)=\brbfun( \brbp,\brbp) \leq \underbrace{\brbfun( \brbb,\brba)}_{=:\brbU},~\forall \brbp\in\brbBox{\brba}{\brbb}
\end{align}
can serve as an upper bound to the best possible solution when $\brbp$ is restricted to a box $\brbB=\brbBox{\brba}{\brbb}=\{\brbp~|~\brba\leq\brbp\leq\brbb\}$.
Even though this is an utopian bound, i.e., there is usually no $\brbp$ for which equality holds in \eqref{eq:utop},
it becomes tight for $\brbb-\brba\to\zero$.

On the other hand, the optimal function value inside a box $\brbB=\brbBox{\brba}{\brbb}$
can be bounded from below by 
the achievable value
\begin{equation}
\label{eq:cbv}
\brbL:=\brbfun( \brba,\brba) 
=\brborigfun(\brba)
\leq \max_{\brbp \in\brbBox{\brba}{\brbb}} \brborigfun(\brbp).
\end{equation}

\subsection{Branch-And-Bound Solution}
The main idea of the branch-and-bound algorithm \cite[Sec.~6.2]{Tu16} is that subdividing a box $\hat\brbB=\brbBoxSmall{\hat\brba}{\hat\brbb}$ 
into a pair of smaller boxes $\brbB_1$ and $\brbB_2$ leads to refined bounds,
which ultimately become tight if the boxes converge to singletons.
The subdivision can be performed using the bisection rule \cite[Sec.~6.2]{Tu16}
\begin{subequations}
\label{eq:branch}
\begin{align}
\brbB_1 &= \brbBox{\hat\brba}{\hat\brbb - \frac{\hatbrbbk[k^\star] -\hatbrbak[k^\star]}{2} \mb e_{k^\star}}
\label{eq:branch1}\\\label{eq:branch2}
\brbB_2 &= \brbBox{\hat\brba +  \frac{\hatbrbbk[k^\star] -\hatbrbak[k^\star]}{2} \mb e_{k^\star}}{\hat\brbb}
\end{align}
\end{subequations}
where $\mb e_k$ is the $k$th canonical unit vector, and
\begin{equation}
\label{eq:branch_k}
k^\star = \argmax_{k\in\{1,2\}} ~~\hatbrbbk-\hatbrbak.
\end{equation}
The intuitive interpretation of this rule is that the box $\hat\brbB=\brbBoxSmall{\hat\brba}{\hat\brbb}$ is cut along its longest edge into two subboxes.

The branch-and-bound algorithm is summarized in Algorithm~\ref{algo:bb}.
For a proof that the procedure converges to an $\brbeps$-optimal solution, see \cite[Sec.~6.2]{Tu16}.
A possible initialization is discussed in the next subsection.
\begin{algorithm}[h]
\caption{Branch-and-Bound Method for Problem~\eqref{eq:inner}}
\label{algo:bb}
Given an initial set $\brbBset=\{\brbB_0\}$ such that the optimizer is contained in the box $\brbB_0$:
\begin{enumerate}
\item Find the box with the highest upper bound, i.e., $\hat\brbB = \argmax_{\brbB\in\brbBset} \brbU[\brbB]$ with $\brbUsymb$ defined in \eqref{eq:utop}.\label{item:brb_find}
\item Replace $\brbBset$ by\footnotemark $(\brbBset\setminus\{\hat\brbB\}) \cup \{\brbB_1,\brbB_2\}$ using \eqref{eq:branch}.\label{item:brb_branch}
\item Repeat Steps~\ref{item:brb_find}) and~\ref{item:brb_branch}) until
$\max_{\brbB\in\brbBset} \brbU[\brbB]-\max_{\brbB\in\brbBset} \brbL[\brbB] \leq \brbeps$ with $\brbLsymb$ defined in \eqref{eq:cbv}.\label{item:brb_converge}
\item Return the vector $\brbp$ that achieves $\max_{\brbB\in\brbBset} \brbL[\brbB]$.
\end{enumerate}
\end{algorithm}
\footnotetext{We use $\setminus$ to denote a set difference.}

\subsection{Initialization}
\label{sec:bb:init}
We introduce $\brbinitfun(\brbp)=\sum_{k=1}^2 \brbinitfun_k(\brbpk)$  with
\begin{equation}
\brbinitfun_k(\brbpk) = 
\mu_k \log_2\left(1+\frac{|\mc h_{kk}|^2 \brbpk }{c_{\noise_k}  }\right) - \lambda_k \brbpk
\end{equation}
where we have neglected the inter-user interference.
This results in an upper bound, i.e., $\brbinitfun(\brbp)\geq\brborigfun(\brbp)$.

The functions $\brbinitfun_k(\brbpk)$ are concave, and they tend to $-\infty$ for large values of $\brbpk$
since the logarithm grows sublinearly.
Due to these properties, it is possible to find 
$\brbinitfun_{\mt{max},k}=\max_{\brbpk\geq 0} \brbinitfun_k(\brbpk)$ for all $k$
by means of convex programming,
and we can apply a simple root finding method to
obtain a value $\brbpk[0,k]$ such that $\brbinitfun_k(\brbpk)+\brbinitfun_{\mt{max},j}\leq0,~\forall \brbpk\geq\brbpk[0,k]$
with $j=3-k$.\footnote{Note that constructing $\brbpk[0,k]$ such that $\brbinitfun_k(\brbpk)\leq0$ for $\brbpk\geq\brbpk[0,k]$ would not be sufficient since this would only guarantee that
$\brbinitfun(\brbp)\leq0$ if $\brbpk\geq\brbpk[0,k]$ holds for all $k$. In the following, we instead need that $\brbinitfun(\brbp)\leq0$ is already guaranteed if $\brbpk\geq\brbpk[0,k]$ holds for some $k$.}

If $\brbpk\geq\brbpk[0,k]$ for any $k$, we have $\brborigfun(\brbp)\leq\brbinitfun(\brbp)\leq0$.
It is thus clear that $\brborigfun$ takes its maximum inside $\brbB_0=\brbBox{\zero}{\brbp_0}$, where $\brbp_0=[\brbpk[0,1],\brbpk[0,2]]^\Tr$.
Therefore, $\brbBset=\{\brbB_0\}$ with $\brbB_0=\brbBox{\zero}{\brbp_0}$ can be used as an initialization for the branch-and-bound method.

\subsection{Related Literature and Remarks on the Complexity}
Algorithms from the field of monotonic optimization have been previously applied in various communication scenarios,
e.g., the polyblock method in
\cite{JoLa08,Br12,BrUt09,JoLa09,JoLa10,QiZhHu09,ZaBjSaJo17}
and the branch-and-bound method in
\cite{HeUtJo11,HeJoRiUt12,GrJoUt12,HeUt12}.

The concept of combining Lagrange duality with a monotonic optimization method for evaluating the dual function was proposed in \cite{Br12,BrUt09} based on the polyblock method,
and adopted in \cite{HeJoRiUt12} using the branch-and-bound method.
A common point of \cite{Br12,BrUt09,HeJoRiUt12} is that a rate space formulation was used, i.e., the per-user rates were used as optimization variables.
This means that an inner loop for finding a feasible transmit strategy has to be executed for each rate vector that the polyblock or branch-and-bound algorithm considers during its execution.
In this paper, we have instead formulated a monotonic optimization problem based on the transmit powers, so that the bounds can be explicitly calculated without an inner iteration.
A similar approach as the one considered here was pursued in \cite{He17} for a multiple-input single-output broadcast channel with TIN.

A particularity of solving problem \eqref{eq:inner} with the branch-and-bound method is that there is no upper bound on the transmit powers
since the power constraints of the original time-sharing problem \eqref{eq:primal} have been dualized.
This effect does not occur when optimizing pure strategies \eqref{eq:noTS} or strategies based on the convex hull formulation \eqref{eq:rts}.
Thus, the initialization could simply be constructed based on the constraint set in most of the papers referenced above.
Intuitively speaking, time-sharing would in principle allow us to use arbitrarily high transmit powers if the respective strategy is used only for a very short fraction of the total time.
This leads to the additional complication of having to find an appropriate initial box based on properties of the objective function instead of based on the constraints (see Section~\ref{sec:bb:init}).
Other examples where such a procedure is necessary can be found in \cite{HeJoRiUt12,He17}.

It needs to be mentioned that the monotonic optimization methods discussed above have in common that they are not adequate for online implementation due to their high computational complexity,
which grows exponentially in the number of optimization variables.
However, due to the non-convex nature of the considered problems, no globally optimal solution methods with lower computational complexity are known.
Therefore, applying monotonic optimization makes sense as a benchmark solution for evaluating the performance of heuristic methods in offline simulations,
or for producing numerical results needed to understand or illustrate fundamental aspects of the considered system.
Moreover, the problem considered in this paper has only two optimization variables (namely the transmit powers of the two users),
so that the overall computational effort remains manageable.

\section{Numerical Results}
\label{sec:num}
To illustrate the statement of Theorem~\ref{th:main} and the application of the proposed optimization method, let us reconsider a numerical example from \cite{HeUt18}.
In contrast to \cite{HeUt18}, where the numerical results were used to study the globally optimal time-sharing solution only under a restriction to proper signaling,
we now know from Theorem~\ref{th:main} that this solution coincides with the globally optimal time-sharing solution for the case where improper signaling is allowed.
This enables us to give much stronger statements in the interpretations of the results.

To allow for a simple comparison with the existing literature, we reconsider the numerical example from \cite[Fig.~3]{ZeYeGuGuZh13},
where the channel realization
\begin{subequations}
\begin{align}
\mc h_{11}&= 2.0310\,\E^{-\J0.6858} &
\mc h_{12}&= 1.4766\,\E^{\J2.6452} \\
\mc h_{21}&= 0.7280\,\E^{\J1.9726} &
\mc h_{22}&= 0.9935\,\E^{-\J0.6676}
\end{align}
\end{subequations}
was assumed with the same transmit power limitation for both users and a signal-to-noise ratio of $10\mt{dB}$ (transmit power over receiver noise).

Under a restriction to pure strategies with proper signaling, it is possible to compute Pareto optimal strategies in a globally optimal manner (see \cite{ZeYeGuGuZh13,LiZhCh12}).
By taking the convex hull of the obtained rate region, the globally optimal rate region for the convex hull formulation with proper signaling can be obtained.
These two rate regions in Fig.~\ref{fig:H1} can also be found in \cite{ZeYeGuGuZh13}.

Unlike for the case of proper signaling, we are not aware of a method to calculate globally optimal rate points for the case of pure strategies with improper signaling.
As an achievable rate region, we therefore use the largest improper signaling rate region given in \cite{ZeYeGuGuZh13}, which can be found by means of a grid search or by random sampling of transmit strategies.
The suboptimal algorithms proposed in \cite{ZeYeGuGuZh13} as well as the rank-one method from \cite{HoJo12} achieve smaller regions.\footnote{Each of the heuristic approaches
discussed in \cite{ZeYeGuGuZh13} performs well for certain choices of the rate profile vector $\mb\rho$, but not for all choices, i.e., not along the whole Pareto boundary.
Thus, we have instead plotted the rate region based on a grid search over all parameters in order to avoid being limited by the suboptimality of the employed heuristic method.}
In Fig.~\ref{fig:H1}, this curve is again accompanied by its convex hull.
\begin{figure}
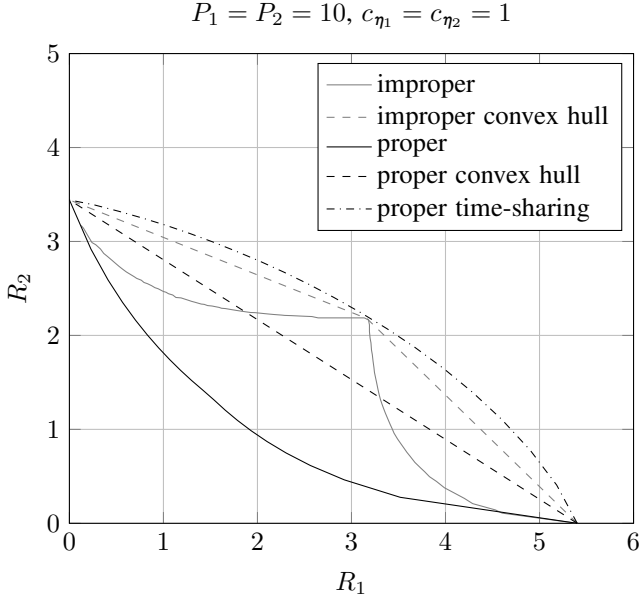

\ifCLASSOPTIONdraftcls\centering\fi
\include{inc_rateregH1}
\caption{Achievable rate regions with pure strategies, with the convex hull formulation, and with time-sharing (scenario from \cite[Fig.~3]{ZeYeGuGuZh13}).}
\label{fig:H1}
\end{figure}

When considering only these four curves, the conclusion is that improper signaling significantly enlarges the rate region compared to proper signaling \cite{ZeYeGuGuZh13}.
However, as shown in Theorem~\ref{th:main}, this result changes completely when considering the possibility of applying time-sharing.

To illustrate this, we have added the curve for proper time-sharing, which can be computed up to an arbitrarily small error tolerance
by means of the cutting plane algorithm described in Section~\ref{sec:algo} combined with the branch-and-bound method from Section~\ref{sec:inner} as a solver for the inner problem.
The remarkable result, which was already observed in \cite{HeUt18}, is that this proper time-sharing solution is not only able to keep up with the improper convex hull solution, but is even superior.
However, due to Theorem~\ref{th:main}, we now obtain an even stronger statement.
Even if we allow the combination of improper signaling and coded time-sharing (instead of taking the convex hull for improper signaling as done in \cite{ZeYeGuGuZh13}),
it is not possible to obtain a larger rate region than with proper signals and coded time-sharing.

\section{Discussion and Outlook}
\label{sec:conclusion}
Using improper signaling instead of proper signaling 
in the two-user Gaussian interference channel (with TIN and Gaussian inputs)
enlarges the rate region achievable with pure strategies 
as well as the convex hull of this rate region.
However, we have shown in this paper that
the situation changes if we consider coded time-sharing based on a time-sharing parameter.
In this case, proper signaling is optimal, i.e., it achieves the whole time-sharing rate region, and improper signaling cannot bring any advantages.

To better understand this result, note that improper signaling gives us additional flexibility compared to proper signaling,
and time-sharing gives us additional flexibility compared to pure strategies and to the convex hull formulation.
It was a~priori not clear whether combining both kinds of flexibility, i.e., improper signaling and time-sharing,
is necessary to achieve the full rate region of the two-user Gaussian interference channel (with TIN and Gaussian inputs).
According to Theorem~\ref{th:main}, this is not necessary.

An intuitive explanation why this is different in case of the convex hull formulation was given in \cite{HeUt17a} for the special case of a one-sided interference channel,
and it analogously applies to the system considered here.
When interpreting improper signaling as separately designing two real-valued data streams for each user,
we observe that reducing the power of one of these streams will give us the freedom to use more power for the other one.
The question is thus whether the same can be achieved by instead averaging between proper signaling strategies over time.
In the convex hull formulation, where the rates are averaged, but not the powers,
reducing the transmit power of a user in one time interval will help the other user by reducing the interference, but it will not allow us to use a higher transmit power in some other interval.
Therefore, using improper signals gives us more flexibility in this respect.
On the other hand, averaging between several strategies over time brings additional flexibility in the sense that the time-sharing weights, which indicate the lengths of the time slots, can be optimized in addition.

When using coded time-sharing based on a time-sharing parameter, we can do both:
we can trade power between time slots and we can arbitrarily adapt the time-sharing weights.
This apparently provides us with enough flexibility to achieve the whole TIN rate region,
and adding further flexibility by allowing improper signals does not bring any advantages.

However, note that this simplified intuitive argumentation was based on transmit powers only,
i.e., it does not consider the possibility of introducing correlations between the real and imaginary parts of the transmit symbols,
which corresponds to adapting the complex phase of the pseudovariances.
Therefore, this argumentation only applies in the case of a two-user system, for which it has been shown in this paper that the phases of the pseudovariances can be chosen such that all impropriety
corresponds to power imbalances between the real and imaginary parts in an equivalent system.
For three or more users, this does not apply since $|\mc{\rtilde c}_{\mc s_k}|^2$ in the denominator of the second summand of \eqref{eq:rk} is then no longer independent of the phases of the pseudovariances and of the channels.
Therefore, the argumentation used in Lemma~\ref{lem:enhanced} is no longer possible then.

The fact that the phases clearly matter in the $K$-user Gaussian interference channel with $K>2$ was already observed in \cite{CaJa08,CaJaWa10}.
It was then shown in \cite{CaJaWa10} that improper signaling can be necessary to achieve the maximum sum degrees of freedom (DOF) in the three-user Gaussian interference channel.
Since such a DOF study is valid regardless of whether or not the powers may be averaged, the result of \cite{CaJaWa10} implies that 
improper signaling can bring benefits in the three-user Gaussian interference channel even if coded time-sharing is allowed.
This does not contradict the observations in this paper, which are specific to the two-user case as explained above.

As the result in \cite{CaJaWa10} was based on a combination of improper signaling and symbol extensions,
we have also verified that symbol extensions do not help in the two-user Gaussian interference channel
at least in the considered case of constant channel coefficients.
An extension to time-varying channel coefficients could be considered in future research.

Further topics that should be considered are whether similar results as in this paper can also be obtained for
MISO, single-input multiple-output (SIMO), and MIMO interference channels with two users.
For all these scenarios, gains by improper signaling can be obtained when pure strategies or the convex hull formulation are considered,
but the combination of improper signaling and coded time-sharing has not yet been considered in these scenarios.
Moreover, extensions to coding schemes other than TIN should be considered in future research, e.g., to the whole Han-Kobayashi rate region.
Finally, studying the interplay of improper signaling and time-sharing is interesting also for discrete alphabets used in practical systems.

\appendix

We provide formal proofs of Lemmas~\ref{lem:enhanced} and~\ref{lem:proper_opt} taking into account the possibility of symbol extensions,
where $T$ channel uses are combined to a single transmission of a $T$-dimensional extended symbol vector (see, e.g., \cite{JaSh08,CaJa08,CaJaWa10,BeLiNaYa16,NaNg19}).
Note that we only consider scenarios with constant channel coefficients in this paper.

\begin{proof}[Proof of Lemma~\ref{lem:enhanced}]
The achievable rate with possibly improper signals and symbol extensions can be expressed in an equivalent $2T\times 2T$ real-valued MIMO system \cite{HeUt15}
with transmit covariance matrices $\mb Q_1$ and $\mb Q_2$. We have \cite{YeBl03}
\begin{align}
r_k &= \frac{1}{2T} \log_2 \frac{\det(\frac{c_{\noise_k}}{2}\id_{2T} + \mb H_{kk} \mb Q_k \mb H_{kk}^\Tr + \mb H_{kj} \mb Q_j \mb H_{kj}^\Tr )}
{\det(\frac{c_{\noise_k}}{2}\id_{2T} + \mb H_{kj} \mb Q_j  \mb H_{kj}^\Tr) }
\end{align}
and the channel matrices can be written in analogy to \eqref{eq:crrmat} as
\begin{align}
\label{eq:MIMOR_chan}
\mb H_{ij} = \underbrace{\id_T \otimes  \begin{bmatrix}\cos \theta_{ij} & -\sin\theta_{ij}\\ \sin\theta_{ij} & \cos\theta_{ij}\end{bmatrix}}_{=:\mb U_{ij}}\underbrace{|\mc h_{ij}|}_{=:\alpha_{ij}} 
\end{align}
where $\otimes$ is the Kronecker product.
Note that $\mb U_{ij}^{-1}=\mb U_{ij}^\bTr$.
Using the eigenvalue decomposition $\mb Q_k = \mb V_k \mb \Phi_k \mb V_k^\bTr,~\allk$, we then have 
\ifCLASSOPTIONdraftcls\else\begin{figure*}[t!]\fi
\begin{subequations}
\label{eq:MIMOR}
\begin{align}
\label{eq:MIMOR1}
r_k &= \frac{1}{2T} \log_2 \frac{\det(\frac{c_{\noise_k}}{2}\id_{2T} + \mb H_{kk} \mb V_k \mb \Phi_k \mb V_k^\bTr \mb H_{kk}^\bTr + \mb H_{kj} \mb V_j \mb \Phi_j \mb V_j^\bTr \mb H_{kj}^\bTr )}
{\det(\frac{c_{\noise_k}}{2}\id_{2T} + \mb H_{kj} \mb V_j \mb \Phi_j \mb V_j^\bTr \mb H_{kj}) } \\
\label{eq:MIMOR2}
&=\frac{1}{2T}\log_2 \frac{\det(\frac{c_{\noise_k}}{2}\id_{2T} + |\alpha_{kj}|^2 \mb \Phi_j + |\alpha_{kk}|^2 \mb V_j^\bTr \mb U_{kj}^\bTr \mb U_{kk} \mb V_k \mb \Phi_k \mb V_k^\bTr\mb U_{kk}^\bTr \mb U_{kj} \mb V_j)}
{\det(\frac{c_{\noise_k}}{2}\id_{2T} + |\alpha_{kj}|^2\mb V_j^\bTr \mb U_{kj}^\bTr \mb U_{kj} \mb V_{j} \mb \Phi_j) } \\
\label{eq:MIMOR3}
&=\frac{1}{2T}\log_2 \frac{\det(\mb D_k)\det(\id_{2T} + |\alpha_{kk}|^2 \mb D_k^{-1} \mb W_k \mb \Phi_k \mb W_k^\bTr )}{\det(\mb D_k)}
\leq\frac{1}{2T}\log_2 \det(\id_{2T} + |\alpha_{kk}|^2 \mb D_k^{-1} \mb {\tilde \Phi}_k) =: {\bar r}_k
\end{align}
\end{subequations}
\ifCLASSOPTIONdraftcls\else\end{figure*}
\eqref{eq:MIMOR} on the top of page~\pageref{eq:MIMOR} \fi
where
\begin{align}
\mb W_k &= \mb V_j^\bTr \mb U_{kj}^\bTr \mb U_{kk} \mb V_k, &
\mb D_k &= \frac{c_{\noise_k}}{2}\id_{2T}+|\alpha_{kj}|^2\mb \Phi_j 
\end{align}
and the diagonal matrix $\mb {\tilde \Phi}_k$ is a reordered version of $\mb \Phi_k$ that is arranged in a way that the $i$th largest entry of $\mb {\tilde \Phi}_k$
is at the same position as the $i$th smallest entry of $\mb D_k$.
The bound is due to the Hadamard inequality \cite[Sec.~7.8]{HoJo13} and due to the optimal ordering of $\mb {\tilde \Phi}_k$,
which can be shown in analogy to the optimality of channel pairing in the relay scenario in \cite{HaWi07}.\footnote{The main argument can be summarized as follows.
Let $x_1\geq x_2\geq0$, $y_1\geq y_2>0$, and $a>0$. Then,
$\log(1+a\frac{x_1}{y_1})+\log(1+a\frac{x_2}{y_2})
\leq
\log(1+a\frac{x_1}{y_2})+\log(1+a\frac{x_2}{y_1})$
is equivalent to
$1+a\frac{x_1}{y_1}+a\frac{x_2}{y_2}+a^2\frac{x_1x_2}{y_1y_2}
\leq
1+a\frac{x_1}{y_2}+a\frac{x_2}{y_1}+a^2\frac{x_1x_2}{y_2y_1}$
$\Leftrightarrow$
$(x_1-x_2)(\frac{1}{y_1}-\frac{1}{y_2})\leq 0$,
which is fulfilled.}

This upper bound ${\bar r}_k$
does not depend on $\mb U_{ij}$ or on $\mb V_k$,
and it is achievable for both users simultaneously if
we can find $\mb V_1$ and $\mb V_2$ such that
$\mb V_2^\bTr\mb U_{12}^\bTr \mb U_{11} \mb V_1=\id_{2T}= \mb V_1^\bTr\mb U_{21}^\bTr \mb U_{22} \mb V_2$.
This is the case if 
\begin{equation}
\label{eq:MIMOR_condition}
\mb U_{21}^\bTr \mb U_{22}\mb U_{21}^\bTr \mb U_{11}=\id_{2T}
\end{equation}
which is fulfilled for the real-valued channel coefficients in \eqref{eq:modelUB}, 
but not necessarily for the original system in \eqref{eq:model}.\footnote{Optimality of diagonal covariance matrices with anti-aligned entries (one increasing, the other decreasing) was
shown in \cite{BeLiNaYa16} for symbol extensions in the real-valued interference channel.
The difference to the equivalent real-valued system considered here is that the channels can be described by scaled identity matrices in \cite{BeLiNaYa16}
while we are facing the special structure given in \eqref{eq:MIMOR_chan}. As a result, the upper bound in \cite{BeLiNaYa16} is always achievable while we need \eqref{eq:MIMOR_condition} to hold for the upper bound in \eqref{eq:MIMOR3} to be achievable.}

As the upper bound ${\bar r}_k$ does not depend on the phases of the channel coefficients in $\mb U_{ij}$, it is equal for the original interference channel \eqref{eq:model}
and for the enhanced interference channel \eqref{eq:modelUB},
but it can be achieved with equality in the enhanced system.
\end{proof}

\begin{proof}[Proof of Lemma~\ref{lem:proper_opt}]
In the proof of Lemma~\ref{lem:enhanced}, we have already shown using \eqref{eq:MIMOR3} that the optimal rates ${\bar r}_k$ in the enhanced system \eqref{eq:modelUB} can be achieved with diagonal covariance matrices
by choosing $\mb V_1=\mb V_2=\id_{2T}$.
We have\footnote{We have substituted the diagonal entry $\tilde \Phi_{k,t}^{(\ms)}$ by $\Phi_{k,t}^{(\ms)}$ since the ordering will implicitly be optimized when optimizing the diagonal entries of $\mb \Phi_k^{(\ms)}$.}
\begin{equation}
\label{eq:rate_2T}
{\bar r}_k^{(\ms)} = \frac{1}{2T} \sum_{t=1}^{2T} \log_2 \left(1 + \frac{|\alpha_{kk}|^2 \Phi_{k,t}^{(\ms)}}{\frac{c_{\noise_k}}{2} + |\alpha_{kj}|^2 \Phi_{j,t}^{(\ms)}} \right) 
\end{equation}
for the rate in the $\ms$th strategy,
and the power constraint \eqref{eq:primal:pow} has to be replaced by
\begin{equation}
\label{eq:pow_2T}
\sum_{\ms=1}^\Ms \tau_\ms \frac{1}{T} \tr{\mb \Phi_k^{(\ms)} }\leq P_k, ~~\allk.
\end{equation}

Using a similar argument as in \cite{NaNg19}, we can use $L'=2TL$ time slots with $\tau_{\ms}'=\tau_{\lceil \frac{\ms}{2T} \rceil}/(2T)$ in \eqref{eq:primal},
and we can then set
\begin{align}
\Phi_{k,t}^{\prime(2T(\ms-1)+s)} = \Phi_{k,s}^{(\ms)}, \quad \forall t\in\{1,\dots,2T\}
\end{align}
for all $s\in\{1,\dots,2T\}$, $\ms\in\{1,\dots,L\}$, and $k\in\{1,2\}$.
This does not change the value on the left hand side of \eqref{eq:pow_2T},
and using \eqref{eq:rate_2T} as rate expression, the left hand side of \eqref{eq:primal:rate} remains unchanged as well.
Thus, there always exists an optimal solution with scaled identity matrices as covariance matrices.
Taking \eqref{eq:crrcov} into consideration, this corresponds to a strategy without symbol extensions and vanishing pseudovariances in each time slot.
\end{proof}

\bibliographystyle{IEEEtran}
\bibliography{ConfIEEE,IEEEabrv,literature}

\renewenvironment{IEEEbiography}[1]
  {\IEEEbiographynophoto{#1}}
  {\endIEEEbiographynophoto}
  
\begin{biography}{Christoph Hellings}
(Member, IEEE) received the B.Sc., Dipl.-Ing., and Dr.-Ing. degrees in electrical engineering (all with high distinction) from the Technical University of Munich (TUM) in 2008, 2010, and 2017, respectively.

He was a post-doctoral researcher and lecturer with the Methods for Signal Processing group, TUM, where his research focused on transmission strategies in communication systems with interference, including concepts such as improper signaling and transmission schemes for multicarrier systems. Moreover, his recent research included the application of machine-learning techniques in communication systems and the topic of energy-efficient communications. He taught courses on information transmission, data processing, machine learning, and the mathematical foundations of signal processing. In 2017, he was a guest lecturer with the Singapore Institute of Technology (SIT). He has recently joined the Quantum Device Lab, ETH Zürich, where his research includes the design and implementation of algorithms for quantum computation and quantum simulation as well as signal processing aspects in superconducting circuits.

Dr. Hellings received the Rohde \& Schwarz Award for his Ph.D. thesis in 2018 and an outstanding teaching assistant award from the student representatives of the TUM Department of Electrical and Computer Engineering in 2011. From 2012 to 2016, he was honored multiple times as an Exemplary Reviewer by the IEEE Communications Society. While being a student, he held a scholarship of the Max Weber Program of the Bavarian state, and for his Diploma thesis, he received an award from the German Association for Electrical, Electronic \& Information Technologies (VDE).
\end{biography}%

\begin{biography}{Wolfgang Utschick}
Wolfgang Utschick (Senior Member, IEEE) completed several years of industrial training programs before he received the Diploma degree in 1993 and the doctoral degree in 1998 in electrical engineering (both with distinction) with a dissertation on machine learning, from Technische Universit{\"a} M{\"u}nchen (TUM), M{\"u}nchen, Germany. Since 2002, he is a Professor at TUM where he is chairing the Professorship of Signal Processing. He teaches courses on signal processing, stochastic processes, optimization theory and machine learning in the field of wireless communications, various application areas of signal processing, and power transmission systems. Since 2011, he has been a regular Guest Professor at Singapore’s new autonomous university, Singapore Institute of Technology, and since 2017 he is serving as the Dean of the Department for Electrical and Computer Engineering, TUM. He holds several patents in the field of multiantenna signal processing and has authored and coauthored a large number of technical articles in international journals and conference proceedings and has been awarded with a couple of best paper awards. He edited several books and is Founder and Editor of the Springer book series Foundations in Signal Processing, Communications and Networking. He has been a Principal Investigator in multiple research projects funded by the German Research Fund (DFG) and coordinated a five years German DFG priority program on Communications Over Interference Limited Networks. He is a member of the VDE and therein a member of the Expert Group 5.1 for Information and System Theory of the German Information Technology Society. He is currently chairing the German Signal Processing Section. He also had been serving as an Associate Editor for IEEE TRANSACTIONS ON SIGNAL PROCESSING and had been member of the IEEE Signal Processing Society Technical Committee on Signal Processing for Communications and Networking.
\end{biography}%

\end{document}

%% file: inc_rateregH1.tex
\hspace*{-4mm}
\begin{tikzpicture}

\begin{axis}[%
view={0}{90},
width=7.5cm,
height=6.25cm,
scale only axis,
xmin=0, xmax=6,
ymin=0, ymax=5,
xmajorgrids,ymajorgrids,
xlabel={$R_1$},
ylabel={\parbox{2cm}{\centering$R_2$\vspace*{-1.3cm}}},
title={$P_1=P_2=10$, $c_{\mcg h_1}=c_{\mcg h_2}=1$},
%
legend style={nodes=right}]

\addplot [
color=gray,
]
coordinates{
(5.40086611903573,0)(4.60266107126201,0.108337703374149)(4.53073398098703,0.142589604089484)(4.28871967108039,0.209640708312131)(4.20128820717661,0.265005514565712)(4.08801312669288,0.325572545275646)(3.98773612479149,0.378945728046728)(3.90422510476794,0.445951135114719)(3.8298769372171,0.496131721896257)(3.77844143351608,0.549084556042214)(3.73494579036219,0.589657759223288)(3.67436303175174,0.648495160749187)(3.62432732967087,0.712811988129783)(3.58890881563732,0.756511157327951)(3.55996232728577,0.798125994197855)(3.52653123149443,0.845297530716772)(3.48987312548718,0.897903738933411)(3.45268761152803,0.959404859658112)(3.43283017297167,1.00634724109307)(3.41138854704196,1.05305400863292)(3.38324248451187,1.11915546592205)(3.36051560960467,1.17305011092053)(3.34089666158138,1.22389489221609)(3.3237620314537,1.27139060374019)(3.30866389835419,1.31595130850952)(3.29501874497233,1.36754303625491)(3.28226082221646,1.42373845174281)(3.26706444064929,1.49882276548535)(3.25828653418821,1.53971053387636)(3.24443679672915,1.60623531287597)(3.23888898416865,1.65790066113023)(3.22942144877997,1.72662141324517)(3.22160228025766,1.79571559158396)(3.21287246435681,1.84417467588282)(3.20736939587985,1.92020973920951)(3.19796659159473,1.9709704605452)(3.19111896399775,2.04019802134161)(3.19111715359059,2.11192285885371)(3.18097681973327,2.16512927968183)(3.1358352373709,2.18663088626968)(3.02024834754187,2.1866394387415)(2.93142195658696,2.18664682849944)(2.83010095169838,2.18665620372852)(2.77355006595268,2.18666190395694)(2.6451118214402,2.18667618146044)(2.57964393410746,2.20076421357943)(2.50178248748547,2.20402497421213)(2.45945510431047,2.20592201033814)(2.366531036127,2.21041260536442)(2.2888675848961,2.2131485757997)(2.25984135612426,2.21615579153734)(2.17054553853398,2.22300058785966)(2.1263597070576,2.22702362453572)(2.07928965960931,2.23152643209454)(1.97479410795155,2.24236165132271)(1.92818553790554,2.2437744230854)(1.88196636159556,2.25019827058532)(1.83258179197168,2.25741655499581)(1.7795661684491,2.26558680082871)(1.72234285921209,2.27491126812548)(1.66018593425109,2.28565410902051)(1.61859060103083,2.28889110745972)(1.59216371472022,2.29816711186957)(1.51705283003301,2.31292953098611)(1.46938164116708,2.31778548932527)(1.43320211260175,2.33061195337702)(1.40174923544682,2.33392973365765)(1.33830649603622,2.35218315335676)(1.29839466305131,2.36296702788277)(1.26616480094979,2.37390156062489)(1.2189088179567,2.38918008184286)(1.18743897004335,2.39944929858821)(1.13302164464047,2.40381111828027)(1.09905658155698,2.43024815001373)(1.05312036139906,2.4441334940688)(1.027574888089,2.45954617630507)(0.998313084680911,2.46814178846959)(0.943967523334548,2.49675857720712)(0.908749062642385,2.5084544330616)(0.874488021244229,2.52926194782772)(0.830154201297225,2.55067988240081)(0.792948507885537,2.56360469313258)(0.765274953001493,2.58745869009332)(0.715404693124106,2.60718312595535)(0.689687137008973,2.6246449157971)(0.644847644753346,2.65183020747022)(0.606120557884038,2.67527298175866)(0.569134085995781,2.70274300507607)(0.529529860995311,2.73373200321373)(0.482442528694604,2.77400366844612)(0.451997239455957,2.80048006164333)(0.405703191374886,2.83973606644378)(0.364933593624962,2.87299099059772)(0.330339544064853,2.91998167962033)(0.280473252364763,2.96514362319881)(0.239942205209896,2.98853145588448)(0.197281573566926,3.059847844224)(0.149692553919725,3.14152036970126)(0.124651697453322,3.17001842304869)(0.124651697453322,3.17001842304869)(0.124651697453322,3.17001842304869) 
};

\addlegendentry{improper};

\addplot [
color=gray,
dashed
]
coordinates{
 (5.40085994659384,0)
 (3.16557756935338,2.1783993682412)(3.15415980886179,2.1831479276983)
 (0,3.44233918628745) 
};
\addlegendentry{improper convex hull};

\addplot [
color=black,
solid
]
coordinates{
 (5.40107596186976,0)(3.51825427595847,0.276891111131963)(2.92293730603897,0.462946630888997)(2.56302517564847,0.615326700476518)(2.30460085377627,0.748808967514082)(2.10181107379327,0.870597973220927)(1.93337418590764,0.98510288530766)(1.78761336282886,1.09545087386504)(1.65736212702736,1.204144085184)(1.53782082823529,1.31342389452243)(1.41831003656675,1.41831076785091)(1.29897255169928,1.52090199370718)(1.18275071283173,1.62792236433706)(1.06737393708362,1.74180473001137)(0.950602589831888,1.86567494092491)(0.830009189968913,2.00382868443439)(0.702675404651716,2.16262902735038)(0.56472769997794,2.35228754314096)(0.410384062901026,2.59110372102359)(0.229621070358315,2.9176903439978)(0,3.44236388446505) 
};
\addlegendentry{proper};

\addplot [
color=black,
dashed
]
coordinates{
 (5.40107596186976,0)(0,3.44236388446505) 
};
\addlegendentry{proper convex hull};

\addplot [
color=black,
dashdotted
]
coordinates{
 (5.40086202272847,0)(5.18210352137675,0.40784039228545)(4.90220021364098,0.776432237630412)(4.59951106371057,1.10424490891398)(4.28896472115683,1.39356911440881)(3.97917803591603,1.648229509616)(3.67501401360404,1.87251316722643)(3.3790546085659,2.07068732729902)(3.0918824631928,2.24638410102727)(2.81405187650699,2.40342735532415)(2.54494936027933,2.5449493600405)(2.28276878490375,2.67277884130697)(2.02654330964932,2.78929757239174)(1.77441703267137,2.89558542741268)(1.52589255591065,2.99473276021665)(1.27820099067067,3.08585016632009)(1.03026000382954,3.17081425205897)(0.779883366819783,3.24844800804947)(0.52599717282696,3.32101543978591)(0.26669930084895,3.3887358945948)(3.81206930803528e-09,3.44228843327902) 
};
\addlegendentry{proper time-sharing};

\end{axis}
\end{tikzpicture}%